\documentclass[11pt,letterpaper,twocolumn]{article}

\usepackage[margin=0.85in]{geometry}
\usepackage{cmap}
\usepackage{times}
\usepackage{amsmath}
\usepackage{amssymb}
\usepackage{amsthm}
\usepackage{mathtools}
\usepackage{enumitem}
\usepackage{booktabs}
\usepackage{array}
\usepackage{multirow}
\usepackage[expansion=false]{microtype}
\usepackage{xcolor}
\PassOptionsToPackage{hyphens}{url}
\usepackage[colorlinks=true,linkcolor=blue,citecolor=blue,urlcolor=blue]{hyperref}
\usepackage{fancyhdr}
\usepackage{titlesec}
\usepackage{natbib}
\usepackage{tikz}
\usetikzlibrary{positioning, arrows.meta, calc}
\usepackage{algorithm}
\usepackage{algpseudocode}

\titleformat*{\section}{\large\bfseries}
\titleformat*{\subsection}{\normalsize\bfseries\itshape}
\titleformat*{\subsubsection}{\small\bfseries}
\titleformat{\paragraph}[runin]{\bfseries}{}{0pt}{}

\theoremstyle{plain}
\newtheorem{theorem}{Theorem}
\newtheorem{lemma}{Lemma}
\theoremstyle{definition}
\newtheorem{definition}{Definition}
\theoremstyle{remark}
\newtheorem*{interpretation}{Interpretation}
\newtheorem*{proofsketch}{Proof sketch}
\newtheorem*{remarkbox}{Remark}

\newcommand{\PoEop}{\mathrm{PoE}}
\newcommand{\negl}{\mathsf{negl}}
\newcommand{\Sig}{\mathsf{Ver}}
\newcommand{\Kgw}{K_{\mathrm{gw}}}
\newcommand{\Krec}{K_{\mathrm{rec}}}
\newcommand{\Fresh}{\mathsf{Fresh}}
\newcommand{\CR}{\mathcal{CR}}
\newcommand{\RL}{\mathcal{RL}}
\newcommand{\KR}{\mathcal{KR}}
\newcommand{\WF}{\mathsf{WF}}
\newcommand{\Adv}{\mathsf{Adv}}
\DeclareMathOperator{\Replay}{Replay}
\DeclareMathOperator{\Sorttopo}{Sort_{topo}}
\DeclareMathOperator{\Lin}{Lin}

\DeclareMathOperator{\Commit}{Commit}
\DeclareMathOperator{\Observe}{Observe}
\DeclareMathOperator{\Seal}{Seal}

\DeclareMathOperator{\Norm}{Norm}
\DeclareMathOperator{\Serialize}{Serialize}
\DeclareMathOperator{\Sign}{Sign}

\DeclareMathOperator{\RevokeEAC}{RevokeEAC}
\DeclareMathOperator{\EACStatus}{EACStatus}

\setlist{itemsep=2pt,topsep=3pt,parsep=0pt,leftmargin=*}

\title{\vspace{-2em}\textbf{Proof of Execution:}\\[4pt]
{\Large Runtime Verification for Governed AI Agent Actions}}

\author{James Rhodes \qquad George Kang \\[2pt]
\normalsize AlphaBitCore, Inc. \\
\normalsize \texttt{\{james,george\}@alphabitcore.com}}

\date{April 2026}

\begin{document}

\twocolumn[
\begin{@twocolumnfalse}
\maketitle
\thispagestyle{fancy}
\begin{abstract}
\noindent
Agent systems increasingly execute rather than advise. When an AI
agent queries regulated data, invokes effectful tools, and mutates
persistent state, correctness is no longer captured by whether a
terminal output looks plausible. The operative questions are
whether each step was authorized under a contract, whether the
recorded history is tamper-evident, and whether the trajectory can
be reconstructed deterministically. We formalize this as
\emph{runtime proof of execution}. An execution is a triple
$x = (C, T, R)$: a contract $C$, an Execution Causal Event Stream
(ECES) $T$, and a replay context $R$. A well-formedness predicate
and five validator-checkable invariants form the PoE validity
predicate
\[
\PoEop(C, T, R) = 1 \iff
\WF(C,T,R) \wedge I_1 \wedge I_2 \wedge I_3 \wedge I_4 \wedge I_{5a}.
\]
Five semantic guarantees $G_1,\ldots,G_5$ describe what we want of
the deployed system: authorization, path compliance, null effect on
deny, history integrity, and replayability. We state soundness
under explicit assumptions --- EUF-CMA signature unforgeability,
hash collision resistance, single-logical-Gateway integrity,
Recorder integrity, trace completeness, dependency-declaration
completeness, and recorder-clock monotonicity --- and prove that
any probabilistic polynomial-time adversary that produces a
PoE-valid execution violating a semantic guarantee yields either a
signature forgery, a hash collision, or a quantified
deployment-failure event. The \emph{Prime Execution Model} (PEM)
separates planning, enforcement, effect, and recordkeeping into
distinct authority planes; a lemma reduces trace completeness to
Effector-exclusive credentialing. An \emph{Execution Attestation
Certificate} (EAC) is issued only when $\PoEop = 1$. In a
preliminary single-node TypeScript prototype, PoE adds
$\approx 2.7$\,ms on a minimal flow and 4.4\% overhead on
concurrent batch workloads. A standard eight-event trace
compresses to $\approx 1.1$\,KB, and injected Gateway-bypass and
trace-mutation attacks are rejected by invariant checks in our
test set. PoE
does not replace consensus, TEEs, or zkVMs; it binds authorization,
effect, history, and replay into a single runtime-checkable object
so that governed execution becomes attestable under contract.

\end{abstract}
\vspace{1em}
\end{@twocolumnfalse}
]

\section{Introduction}

Agent systems increasingly execute rather than advise. An AI agent
can compose a plan, call external tools, apply policy filters, and
produce persistent effects without a human intermediary at each
step. When execution is compliance-sensitive, the binding question
is not whether a terminal output looks correct; it is whether each
step was authorized, whether the system stayed within enforced
policy, whether durable mutations correspond to permitted actions,
and whether the trajectory can be reconstructed. These are
questions about the execution \emph{trajectory}, not its terminal
artifact.

Mainstream architectures answer them only partially. APIs and
microservices execute without proving authorization bounds.
Distributed tracing systems such as
OpenTelemetry~\citep{opentelemetry} observe execution for
debugging, not for authorization or replay.
Consensus~\cite{castro1999pbft} agrees on state transitions, not
on the off-chain trajectories that produced them.
TEEs~\cite{costan2016sgx} attest to code integrity inside an
enclave; an attested enclave can still take an unauthorized
sequence of actions. Verifiable
computation~\cite{gennaro2010vc} and zkVMs prove that a specified
function was computed correctly but do not bind authorization,
enforcement path, durable effect, or replay context into the
same object. Supply-chain attestation
frameworks~\cite{torres2019intoto, slsa} attest to how a build
artifact was produced, not how a live agent execution unfolded.

The gap is live in regulated settings. FINRA's 2026 Annual Regulatory
Oversight Report, released December~9,~2025, adds a new section
on generative AI that, for the first time, articulates supervision,
recordkeeping, testing, and ongoing-monitoring considerations for
agentic systems capable of autonomously executing multi-step
tasks~\cite{finra2026}.
Contemporary legal analysis reads this development as a
redirection of regulatory focus from content to
conduct~\cite{snell2025}. Logging alone is insufficient when the
system under supervision is itself adaptive.

\subsection{Motivating domain}
\label{sec:motivation}

In regulated industries, supervisory obligations increasingly
apply to the \emph{conduct} of an AI system, not only its final
output~\cite{finra2026, snell2025}. A firm must be able to answer
whether the system acted under an approved authority, whether it
stayed within the approved path, whether denied actions were
side-effect free, whether records are complete, and whether the
full chain can be reconstructed after the fact. These obligations
map directly onto PoE's guarantee structure
(Table~\ref{tab:regmap}). PoE does not replace supervision; it
provides a verifiable runtime substrate for it.

\paragraph{Why LLM planners specifically.}
An LLM planner is architecturally unlike traditional control
software. Its behavior can be corrupted by inputs it reads --- a
retrieved document, a tool output, a conversational
turn~\cite{greshake2023injection} --- without any observable
change to the planner itself. The compromise is undetectable at
the planner boundary, and the standard software-engineering
defenses (unit tests, code review, fuzzing) do not transfer to
prompt-space adversaries. Containment must therefore happen
\emph{below} the planner: at the interface between proposed
actions and durable effects. PoE places that interface on an
explicit, signed, replayable record. The architectural weight
follows from the trust model of the component being bounded, not
from the generic desirability of observability.

\begin{table}[!t]
\centering\small
\renewcommand{\arraystretch}{1.15}
\begin{tabular}{@{}p{1.05in} p{0.65in} p{1.35in}@{}}
\toprule
\textbf{Regulatory risk} & \textbf{PoE map} & \textbf{Technical link} \\
\midrule
Supervisory substitution (Rules 3110 / 3120) &
$G_1, G_2$; $I_1, I_2$ &
Each step carries contract- and Gateway-level authorization. \\
Books-and-records integrity (Rule 4511 / SEC 17a-4) &
$G_4$; $I_4$ &
ECES preserves the full causal chain, not only the terminal output. \\
Audit and replay &
$G_5$; $I_{5a}$, Thm.~\ref{thm:replay} &
Audit shifts from log interpretation to deterministic replay. \\
Scope and authority &
$G_1, G_2$; $I_1, I_2$ &
Scope escape cannot become a valid PoE. \\
\bottomrule
\end{tabular}
\caption{PoE guarantees as a runtime substrate for supervisory
obligations articulated in FINRA's 2026 examination
guidance~\cite{finra2026, snell2025}.}
\label{tab:regmap}
\end{table}

\subsection{Contributions}

This paper develops \emph{Proof of Execution} (PoE), a runtime
verification architecture for governed AI-agent actions. Our
contributions are:

\begin{itemize}
\item A \textbf{formal model} of governed execution as a
proof-carrying object $x = (C, T, R)$, with an event schema, a
replay context, and a well-formedness predicate (\S\ref{sec:model}).
\item A \textbf{revised invariant set}
$\{I_1, I_2, I_3, I_4, I_{5a}\}$ capturing contract freshness,
Gateway scope, deny-side null effect, sealed commit order, and
envelope closure (\S\ref{sec:invariants}).
\item \textbf{Three formal results}: invariant minimality by
witness construction (Thm.~\ref{thm:independence}); soundness
under an explicit PoE-forgery game with advantage bounded by
EUF-CMA signature security, collision resistance, and quantified
deployment-failure terms $\epsilon_{\text{tc}}, \epsilon_{\text{dep}},
\epsilon_{\text{clock}}$ (Thm.~\ref{thm:safety}); and
deterministic replay under interpreter determinism and dependency
declaration (Thm.~\ref{thm:replay}).
\item A \textbf{validator} (Alg.~\ref{alg:validate}) making the
invariant set operational on a causal event DAG.
\item An \textbf{architecture}, the \emph{Prime Execution Model}
(PEM, \S\ref{sec:pem}), separating planning, enforcement, effect,
and recordkeeping to prevent privilege collapse.
\item An \textbf{attestation protocol} issuing an Execution
Attestation Certificate (EAC) only when $\PoEop = 1$
(\S\ref{sec:attestation}).
\end{itemize}

PoE's contribution is not to prove arbitrary computation (as a
zkVM does) or to agree on state (as consensus does). It is to
bind authorization, enforcement path, durable effect,
tamper-evident history, and replay context into a single
runtime-checkable object. The individual invariants have
antecedents in the runtime-verification, reference-monitor, and
tamper-evident-logging literatures; the paper's contribution is
fourfold: (i)~the \emph{composition} of these properties into one
proof-carrying object; (ii)~an explicit separation of
cryptographic assumptions from deployment assumptions, with the
latter entering the soundness bound as named parameters
$\epsilon_{\text{tc}}, \epsilon_{\text{dep}},
\epsilon_{\text{clock}}$ rather than as hand-waving; (iii)~a
reduction (Lemma~\ref{lem:tc}) showing that the strongest
deployment assumption (trace completeness) follows from a
concrete architectural discipline, Effector-exclusive
credentialing; and (iv)~a mapping from this formal structure to
current supervisory expectations for agentic AI. Work, in this
paradigm, is not computation consumed or state agreed upon; it is
\emph{attestable execution under contract} (Fig.~\ref{fig:loop}).

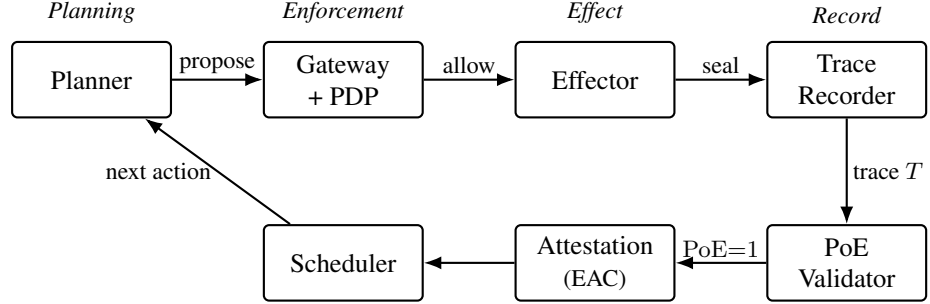
\begin{figure*}[!t]
\centering
\begin{tikzpicture}[
    >=Latex,
    box/.style={
        rectangle, draw, thick, rounded corners=2pt,
        minimum width=2.1cm, minimum height=1.0cm,
        align=center, font=\small
    },
    flow/.style={->, thick},
    lbl/.style={font=\footnotesize, inner sep=2pt}
]
\node[box] (P) {Planner};
\node[box, right=1.2cm of P] (G) {Gateway\\ + PDP};
\node[box, right=1.2cm of G] (E) {Effector};
\node[box, right=1.2cm of E] (R) {Trace\\Recorder};
\node[box, below=1.4cm of R] (V) {PoE\\Validator};
\node[box, left=1.2cm of V] (A) {Attestation\\{\footnotesize (EAC)}};
\node[box, left=1.2cm of A] (S) {Scheduler};
\node[font=\footnotesize\itshape, above=0.1cm of P] {Planning};
\node[font=\footnotesize\itshape, above=0.1cm of G] {Enforcement};
\node[font=\footnotesize\itshape, above=0.1cm of E] {Effect};
\node[font=\footnotesize\itshape, above=0.1cm of R] {Record};
\draw[flow] (P) -- node[above, lbl] {propose} (G);
\draw[flow] (G) -- node[above, lbl] {allow} (E);
\draw[flow] (E) -- node[above, lbl] {seal} (R);
\draw[flow] (R) -- node[right, lbl] {trace $T$} (V);
\draw[flow] (V) -- node[above, lbl] {$\PoEop{=}1$} (A);
\draw[flow] (A) -- (S);
\draw[flow] (S) -- node[left, lbl] {next action} (P);
\end{tikzpicture}
\caption{The PoE closed loop. The top row is a single governed
execution: Planner proposes under contract $C$; Gateway (PEP/PDP)
emits the canonical allow/deny; Effector performs durable
mutation within the authorized scope; Recorder seals the causal
stream $T$. The bottom row is the governance loop: the validator
checks $\PoEop(C,T,R)$; on success an EAC is issued; the
scheduler uses it as authorization for the next action.}
\label{fig:loop}
\end{figure*}

\section{System Model}
\label{sec:model}

\subsection{Execution as a proof-carrying object}

\begin{definition}[Execution]
An execution is a triple $x = (C, T, R)$ where $C$ is a
\emph{contract}, $T$ is an \emph{Execution Causal Event Stream}
(ECES), and $R$ is a \emph{replay context vector}.
\end{definition}

The contract $C$ fixes the execution boundary: principal,
authorized capability set, applicable policy snapshot, validity
window $[t_{\mathrm{nb}}, t_{\mathrm{na}}]$, revocation reference,
and replay configuration. A distinguished identifier
$\mathrm{id}(C) = H(C)$ uniquely names $C$ under a
collision-resistant hash $H$. No effectful execution is valid
without a contract.

\paragraph{Contract registry and revocation.}
A verifier must dereference $\mathrm{id}(C)$ to the full contract
content. The Record Plane maintains a content-addressed
\emph{Contract Registry} $\CR : H(C) \mapsto C$, itself
append-only and tamper-evident. Contract revocation is represented
by a revocation log $\RL$. Define
$\Fresh(C, t, \RL) = 1$ iff no valid revoke event for $C$ appears
in $\RL$ before recorder time $t$. The predicate is
\emph{point-in-time}: revocation at $t{=}5$ does not
retrospectively invalidate an execution that was valid at $t{=}3$.
Whether a downstream workflow accepts an already-issued EAC is a
separate lifecycle question (\S\ref{sec:eac-lifecycle}); we write
\[
\EACStatus(\mathrm{eac\_id}, t) \in
\{\mathsf{active}, \mathsf{revoked}, \mathsf{suspended}\}
\]
and distinguish historical PoE validity from current EAC
acceptability.

The replay context $R$ pins the deterministic environment: planner
version, model or persona hash, policy snapshot hash, memory
snapshot hash, capability binding, tool schemas, dependency
manifest, captured random seeds, and captured external inputs per
event.

\subsection{The proof atom}

\begin{definition}[Event]
\label{def:event}
An event $e \in T$ is the smallest proof-bearing unit of
execution. It is a tuple of fields
\begin{align*}
e = (&\mathrm{id},\ \mathrm{parent},\ \mathrm{contract\_hash},\ \mathrm{principal},\\
& \mathrm{capability},\ \mathrm{authorized\_scope},\ \mathrm{gateway\_ref},\\
& \mathrm{decision},\ \mathrm{input\_hash},\ \mathrm{tool\_schema\_hash},\\
& \mathrm{effect\_type},\ \mathrm{resource\_id},\ \mathrm{delta\_hash},\\
& \mathrm{envelope\_hash},\ \mathrm{prev\_event\_hash},\\
& t_{\mathrm{rec}},\ \mathrm{commit\_seq},\ \sigma),
\end{align*}
with $\mathrm{decision} \in
\{\mathsf{allow}, \mathsf{deny}, \bot\}$ (the latter for
non-Gateway events),
$\mathrm{effect\_type} \in
\{\mathsf{none}, \mathsf{mutation}, \mathsf{external}\}$, and
$\sigma$ a signature over the canonical serialization. An event
is \emph{effectful} iff $\mathrm{effect\_type} \neq \mathsf{none}$.
\end{definition}

Three field groups deserve explicit semantics.
$\mathrm{authorized\_scope}$ is nonempty only on Gateway allow
events; it is a capability set or predicate, and containment is
$g.\mathrm{authorized\_scope}(e.\mathrm{capability}) = 1$.
$\mathrm{gateway\_ref}$ points from each effectful event to its
authoritative Gateway decision. The determinism envelope
$\mathrm{envelope\_hash}$ is computed over canonical serialization
of the replay-critical fields (\S\ref{sec:invariants}).
Wall-clock timestamps are advisory; the authoritative time
witness is $t_{\mathrm{rec}}$, assigned by the Recorder, on which
contract validity checks rely.

\begin{definition}[ECES]
An Execution Causal Event Stream is a tamper-evident causal event
structure $T = (V, E, \prec, \ell)$: $V$ events, $E$ parent
edges, $\prec$ the strict transitive closure of $E$, $\ell$
labeling. ECES is append-only, authority-aware, replay-aware, and
cryptographically sealed. The linear sequence of recorder commits
is a \emph{canonicalization} of the DAG, not the whole semantic
structure.
\end{definition}

Events with no parent path in either direction are
\emph{concurrent} in the causal
sense~\cite{mazurkiewicz1987trace}. To make replay deterministic,
the canonical linearization
$\Lin(T) = \Sorttopo(T, \text{key} = (\mathrm{commit\_seq},
\mathrm{id}))$ respects $\prec$ and breaks ties by recorder
sequence and event id. For concurrent mutations to a shared
resource we require \emph{resource serialization}: for any two
effectful events with the same $\mathrm{resource\_id}$, either
$e_i \prec e_j$ or $e_j \prec e_i$.

\paragraph{Well-formedness.}
\begin{multline*}
\WF(C, T, R) = \WF_{\text{type}} \wedge \WF_{\text{dag}}
\wedge \WF_{\text{contract}}\\
\wedge \WF_{\text{resource}} \wedge \WF_{\text{time}}.
\end{multline*}
The subpredicates assert that fields are well-typed, the parent
graph is acyclic and consistent with $\prec$, every
$\mathrm{contract\_hash}$ dereferences through $\CR$, effectful
events on shared resources are deterministically serializable, and
$t_{\mathrm{rec}}$ and $\mathrm{commit\_seq}$ are monotone.

PoE sits adjacent to tamper-evident audit
logs~\cite{crosby2009secure, laurie2014ct} and provenance
standards~\cite{w3cprov} but differs in scope: those systems
protect records \emph{of} execution; PoE also binds authorization,
effect, and replay into the same object, so the record
reconstructs the execution.

\section{A Worked Example}
\label{sec:example}

We illustrate the formalism with an agent task: \emph{``Retrieve
the two-week price performance for a specified public
equity.''}\footnote{Numeric values in this example are
illustrative; no real market data is asserted.}

On task creation the planner compiles the request into two agent
nodes, $\mathsf{extract\_ticker} \to
\mathsf{retrieve\_performance}$, under a single contract $C$. The
execution-permissions block narrows a requested capability set
from 4 patterns to 13 resolved entries, of which 2 are exercised.
This $4 \to 13 \to 2$ narrowing illustrates $I_2$: the planner
requests a capability family, the governance plane resolves it
under the principal's policy, and only the narrowly authorized
skills are invoked. Four events anchor the trace:
$\mathsf{TASK\_CREATED}$, $\mathsf{EB\_COMPLETE}$ (seals the
compiled contract), $\mathsf{CAPABILITY\_RESULT}$ (records the
price-fetch call), $\mathsf{TASK\_COMPLETED}$ (records the
terminal response). Each event is signed; causal parentage is
explicit.

\paragraph{Scope containment.}
$\mathsf{retrieve\_performance}$ has capability
$e.\mathrm{capability} = \mathsf{web.fetch.market\_price}$. The
Gateway allow event $g$ carries
$g.\mathrm{authorized\_scope} \supseteq
\{\mathsf{web.fetch.market\_price}\}$, so
$g.\mathrm{authorized\_scope}(e.\mathrm{capability}) = 1$.
$\mathrm{Check}_{I_2}$ verifies $g.\mathrm{id} =
e.\mathrm{gateway\_ref}$, checks $\Sig(g, \Kgw) = 1$, and
evaluates scope containment.

\paragraph{Deny branch.}
Consider a hypothetical extension in which the agent attempts
$e'.\mathrm{capability} = \mathsf{brokerage.place\_order}$.
Gateway emits $g'.\mathrm{decision} = \mathsf{deny}$ with
$g'.\mathrm{authorized\_scope} = \emptyset$. The deny event is
sealed, but every descendant must satisfy
$e.\mathrm{effect\_type} = \mathsf{none}$ and
$e.\mathrm{delta\_hash} = \bot$. A trading-state write after the
deny causes $\mathrm{Check}_{I_3}$ to fail; a write that \emph{is
not} recorded is not an exception to $I_3$ but a
trace-completeness failure.

\paragraph{Adversarial mutation.}
Suppose an adversary replaces the HTTP response body in
$\mathsf{CAPABILITY\_RESULT}$ with a different price sequence.
Because $\mathrm{input\_hash}$, $\mathrm{envelope\_hash}$, and
the successor's $\mathrm{prev\_event\_hash}$ bind the original
body, the substitution fails either
$H(\Serialize(\Norm(\mathrm{input}'))) \neq e.\mathrm{input\_hash}$
or $e_{i+1}.\mathrm{prev\_event\_hash} \neq H(e'_i)$.
$\mathrm{Check}_{I_4}$ or $\mathrm{Check}_{I_{5a}}$ rejects.

\paragraph{What replay requires.}
To reconstruct the execution a verifier needs $(C, T, R)$: the
contract resolved through $\CR$, the sealed trace $T$, and the
replay context $R$ containing planner version, tool schemas, the
dependency manifest, and the captured HTTP response bound to the
relevant event. Without the captured response the trace is not
replayable. Replay failure from a missing but declared
dependency is an $I_{5a}$ failure; failure from an undeclared
dependency enters through $\epsilon_{\text{dep}}$ (A6).

\section{Threat Model and Assumptions}
\label{sec:threats}

\subsection{Adversary}

Let $\mathcal{A}$ be a probabilistic polynomial-time adversary
that may submit crafted requests, compromise the planner, observe
and replay network traffic, attempt direct tool invocation outside
the intended path, exploit environmental non-determinism, and
attempt to fabricate or mutate traces. The adversarial goal is to
output $(C^*, T^*, R^*, \omega)$ such that the validator accepts
while $\omega$ witnesses a violation of some semantic guarantee.

\subsection{Assumptions A1--A7}
\label{sec:assumptions}

Let $\lambda$ be the security parameter for signature and hash
primitives. The deployment-failure terms below are not necessarily
functions of $\lambda$; they enter Thm.~\ref{thm:safety} as
independent parameters.

\begin{description}[leftmargin=1.2em,style=nextline]
\item[\textbf{A1 Signature unforgeability.}]
Gateway and Recorder signature schemes are EUF-CMA.
\item[\textbf{A2 Collision resistance.}]
The hash family used for hash-linking and Merkle sealing is
collision-resistant.
\item[\textbf{A3 Single-logical-Gateway integrity.}]
Thm.~\ref{thm:safety} is stated for a single \emph{logical}
Gateway: one canonical decision record per effectful branch. The
Gateway signing root is uncompromised. Replicated or threshold
Gateways require additional quorum assumptions and are future work.
\item[\textbf{A4 Recorder integrity.}]
The Recorder signing key and sealing root remain inside the
trusted computing base. Recorder-assigned $\mathrm{commit\_seq}$
and $t_{\mathrm{rec}}$ are monotone.
\item[\textbf{A5 Trace completeness.}]
For every persistent mutation or external effect $m$ inside the
declared PoE effect boundary there exists $e \in T$ with
$\Observe(e, m) \wedge \Seal(e) < \Commit(m)$. Failure probability
is $\epsilon_{\text{tc}}$.
\item[\textbf{A6 Dependency-declaration completeness.}]
Each transition $\delta_e$ declares its complete dependency set.
An undeclared, uncaptured input occurs with probability
$\epsilon_{\text{dep}}$.
\item[\textbf{A7 Recorder-clock monotonicity.}]
$t_{\mathrm{rec}}$ is monotone within a trace; contract validity
checks rely on $t_{\mathrm{rec}}$, not wall-clock. Failure
probability is $\epsilon_{\text{clock}}$.
\end{description}

A5--A7 are \emph{deployment} assumptions, not cryptographic ones.
They are supported (not implied) by the Prime Execution Model
(\S\ref{sec:pem}), specifically by making the Effector the sole
holder of mutation credentials. We do not claim PoE soundness
exceeds the strength of these deployment controls. The symbols
$\epsilon_{\text{tc}}, \epsilon_{\text{dep}},
\epsilon_{\text{clock}}$ denote deployment-reliability
parameters rather than cryptographic probabilities; they are
bounded by operational controls (credential management, runtime
instrumentation, clock discipline), not by the security parameter
$\lambda$.

\paragraph{Captured inputs in practice.}
For an LLM agent, the captured-input set per effectful event
typically contains tool response bodies, retrieved document
snippets, and any model-facing context that influenced the
transition. Section~\ref{sec:eval} reports that a standard
eight-event execution compresses to $\approx 1.1$\,KB, of which
captured inputs are the dominant component. Replay cost is not
storage-dominated: the binding constraint is that tool
invocations during replay read from captured artifacts rather
than from live external sources that may have drifted. Inputs
that cannot be captured --- e.g.\ time-sensitive external state
consumed outside the declared envelope --- are category~(iii)
dependencies (\S\ref{sec:replay}) and contribute to
$\epsilon_{\text{dep}}$ rather than to $I_{5a}$.

\subsection{Threat classes and scope}

\begin{itemize}
\item[\textbf{T1}] Unauthorized execution.
\item[\textbf{T2}] Gateway bypass.
\item[\textbf{T3}] Deny-with-effect.
\item[\textbf{T4}] Trace mutation or fabrication.
\item[\textbf{T5}] Replay evasion.
\item[\textbf{T6}] Credential escape / trace-completeness failure.
\end{itemize}

Total compromise of the Recorder sealing root, collapse of a
future threshold-Gateway quorum, broken primitives, and physical
side channels outside the declared boundary lie outside the
guarantees of this analysis. Any in-bound persistent mutation
that does not enter $T$ is \emph{not} an out-of-bound exception;
it contributes to $\epsilon_{\text{tc}}$. Prompt injection on the
planner~\cite{greshake2023injection} is not prevented by PoE; PoE
confines what a compromised planner can do within an issued
contract.

\section{Invariants, Guarantees, and the Validator}
\label{sec:invariants}

Each invariant is a syntactic predicate over $(C, T, R)$
decidable in polynomial time in $|T|$.

\paragraph{I1: Zero-trust start with contract freshness.}
There exists a unique root allow event $e_a \in T$ for the
execution under $C$ with $e_a.\mathrm{decision} = \mathsf{allow}$,
$e_a.\mathrm{contract\_hash} = H(C)$,
$\Sig(e_a, \Kgw) = 1$,
$t_{\mathrm{nb}} \leq e_a.t_{\mathrm{rec}} \leq t_{\mathrm{na}}$,
and $\Fresh(C, e_a.t_{\mathrm{rec}}, \RL) = 1$. Every effectful
event $e$ satisfies $e_a \preceq e$. $I_1$ governs
\emph{contract-level} authorization of the execution root.

\paragraph{I2: Gateway evaluation with scope and freshness.}
For every effectful event $e$ there exists a unique Gateway
decision event $g \in T$ with $g.\mathrm{id} =
e.\mathrm{gateway\_ref}$, $g.\mathrm{decision} = \mathsf{allow}$,
$g.\mathrm{contract\_hash} = H(C)$, $g \prec e$,
$g.\mathrm{authorized\_scope}(e.\mathrm{capability}) = 1$,
$\Sig(g, \Kgw) = 1$,
$t_{\mathrm{nb}} \leq e.t_{\mathrm{rec}} \leq t_{\mathrm{na}}$,
and $\Fresh(C, e.t_{\mathrm{rec}}, \RL) = 1$. $I_2$ governs
\emph{action-level} authorization and scope inheritance. Both
invariants are needed: without $I_2$, a well-rooted execution may
escape scope; without $I_1$, locally well-scoped actions lack an
authoritative root.

\paragraph{I3: Null effect on deny.}
For every $d \in T$ with $d.\mathrm{decision} = \mathsf{deny}$,
every causal descendant $e$ in the same branch satisfies
$e.\mathrm{effect\_type} = \mathsf{none}$ and
$e.\mathrm{delta\_hash} = \bot$. A persistent effect on the deny
branch that is \emph{not} recorded is a trace-completeness failure,
not an exception to $I_3$.

\paragraph{I4: Immutability and sealed commit order.}
For each event $e_i$ in recorder commit order,
$e_i.\mathrm{prev\_event\_hash} = H(e_{i-1})$,
$\Sig(e_i, \Krec) = 1$. The sealing commitment verifies;
$\mathrm{commit\_seq}$ is strictly monotone and consistent with
$\Lin(T)$.

\paragraph{I5a: Syntactic envelope closure.}
For every event $e$, $e.\mathrm{envelope\_hash}$ equals the hash
over the canonical concatenation of the replay-critical fields in
$R$ for $e$: planner version, policy snapshot hash, tool schema
(indexed by $e.\mathrm{capability}$), seeds, and captured inputs.
$I_{5a}$ checks only that the declared envelope is consistent
with $R$; discovery of undeclared dependencies is not within its
scope and is covered by A6.

\begin{definition}[PoE validity]
\label{def:poe}
$\PoEop(C, T, R) = 1 \iff \WF(C,T,R) \wedge I_1 \wedge I_2
\wedge I_3 \wedge I_4 \wedge I_{5a}$.
\end{definition}

\paragraph{Semantic guarantees.}
The guarantees $G_1,\ldots,G_5$ are properties of the
\emph{deployed system} and the execution that actually occurred.
\emph{(G1) Authorization:} no effectful action occurred without a
valid allow under $C$. \emph{(G2) Path compliance:} every
effectful action stayed within the scope authorized by a
canonical Gateway evaluation. \emph{(G3) Null effect on deny:}
denied branches produced no persistent state. \emph{(G4) History
integrity:} recorded $T$ is an unaltered record of what occurred.
\emph{(G5) Replayability:} replay under $(C, T, R)$ on a
compliant interpreter yields the same terminal result and
persistent state delta.
Let $\mathcal{G} = G_1 \wedge \cdots \wedge G_5$. The distinction
between $I_k$ (what a validator checks) and $G_k$ (what we want)
separates soundness from definition.

\paragraph{Validator.}
Algorithm~\ref{alg:validate} is a direct transcription. Its cost
is
$\mathcal{O}(|T|\log|T|) +
\mathcal{O}(|T| \cdot \text{Cost}_{\text{sig}})
+ \mathcal{O}(|T| \cdot \text{Cost}_{\text{hash}})$,
where the log factor covers linearization and resource
serialization.

\begin{algorithm}[t]
\caption{$\mathsf{ValidatePoE}(C, T, R)$}
\label{alg:validate}
\begin{algorithmic}[1]
\Require contract $C$; event DAG $T$; replay context $R$;
registry $\CR$; revocation log $\RL$
\Ensure \textsf{valid} / (\textsf{invalid}, $F$)
\State $F \gets \emptyset$
\If{$\neg\,\mathrm{CheckWF}(C, T, R, \CR, \RL)$}
  \State $F \gets F \cup \{\mathsf{WF}\}$
\EndIf
\ForAll{$J \in \{I_1, I_2, I_3, I_4, I_{5a}\}$}
  \If{$\neg\,\mathrm{Check}_J(C, T, R, \CR, \RL)$}
    \State $F \gets F \cup \{J\}$
  \EndIf
\EndFor
\If{$F = \emptyset$} \State \Return \textsf{valid} \EndIf
\State \Return (\textsf{invalid}, $F$)
\end{algorithmic}
\end{algorithm}

\section{Formal Results}
\label{sec:results}

\subsection{Minimality of the invariant set}

\begin{theorem}[Independence under A1--A7]
\label{thm:independence}
For each $k \in \{1, 2, 3, 4\}$, there exists a witness system
satisfying $\WF(C, T, R)$, assumptions A1--A7, and every
invariant in $\{I_1, \ldots, I_{5a}\} \setminus \{I_k\}$, but
violating $G_k$. For $k{=}5$, there exists a witness satisfying
$\WF$ and every invariant (including $I_{5a}$) but violating
$G_5$ through a failure of A6. Consequently, under A1--A7,
removing any invariant from the set breaks the soundness bound
of Theorem~\ref{thm:safety}.
\end{theorem}

\begin{proofsketch}
Appendix~\ref{app:witnesses} gives five witnesses. $S_1$ omits
the root allow (violates $G_1$); $S_2$ bypasses Gateway (violates
$G_2$); $S_3$ records a deny-with-effect chain (violates $G_3$);
$S_4$ rewrites trace payloads (violates $G_4$); $S_5$ reads an
undeclared dependency and diverges on replay (violates $G_5$ via
A6 failure, with $I_{5a}$ passing syntactically). The asymmetry
of $S_5$ reflects a design choice: $I_{5a}$ is a
validator-checkable syntactic proxy for $G_5$ that depends on
dependency-declaration completeness; discovery of undeclared
dependencies from the trace alone is not decidable and therefore
enters Thm.~\ref{thm:safety} through $\epsilon_{\text{dep}}$
rather than through an invariant.
\qedhere
\end{proofsketch}

\subsection{Soundness}
\label{sec:soundness}

\begin{remarkbox}[Logical vs.\ physical Gateway uniqueness]
$I_2$'s "unique Gateway decision" is \emph{logical}: replicated
or quorum deployments~\cite{castro1999pbft} may produce one
canonical decision record through an underlying protocol.
Thm.~\ref{thm:safety} is stated for single-logical-Gateway
deployments under assumption A3; threshold Gateways require
additional quorum assumptions and are future work.
\end{remarkbox}

\begin{definition}[PoE-forgery game]
\label{def:game}
$\mathrm{Game}^{\mathcal{A}}_{\text{poe}}(\lambda)$: the
challenger generates Gateway and Recorder key pairs and
initializes $\CR$, $\RL$, and a hash/Merkle scheme $H_\lambda$;
$\mathcal{A}$ queries oracles $O_{\mathrm{gw}}, O_{\mathrm{rec}},
O_{\mathrm{reg}}, O_{\mathrm{rev}}$ and outputs
$(C^*, T^*, R^*, \omega)$; $\mathcal{A}$ wins iff
$\mathsf{ValidatePoE}(C^*, T^*, R^*) = \textsf{valid}$ and
$\neg G_k(C^*, T^*, R^*, \omega)$ for some $k$. Define
$\Adv^{\text{poe}}_{\mathcal{A}}(\lambda) =
\Pr[\text{win}]$.
\end{definition}

Define basic failure events $F_{\text{gw\_sig}}$
(Gateway signature forgery), $F_{\text{rec\_sig}}$ (Recorder
forgery), $F_{\text{hash}}$ (hash or Merkle collision),
$F_{\text{tc}}$ (trace-completeness failure),
$F_{\text{dep}}$ (dependency-declaration failure),
$F_{\text{clock}}$ (clock-monotonicity failure). Three reductions
extract cryptographic failures:

\paragraph{$B_{\mathrm{gw}}$ (EUF-CMA reduction).}
Embed the EUF-CMA public key as $pk_{\mathrm{gw}}$; route
$O_{\mathrm{gw}}$ queries through the signing oracle over
$\Norm(\mathrm{decision}, \mathrm{scope}, H(C), \mathrm{req})$. If
$\mathcal{A}$ wins, search $T^*$ for a Gateway decision whose
canonical normalization was never signed; output it as the forgery.

\paragraph{$B_{\mathrm{rec}}$.}
Symmetric; covers fabricated events, modified payloads, and
reordered seal roots.

\paragraph{$B_H$.}
Honestly simulate all oracles; on any winning transcript inspect
the hash bindings ($H(C)$, $H(e)$, $\mathrm{input\_hash}$,
$\mathrm{delta\_hash}$, $\mathrm{envelope\_hash}$, Merkle roots)
for a collision pair.

\begin{theorem}[Soundness]
\label{thm:safety}
Let $\mathcal{A}$ be any PPT adversary. Under A1--A2 and the
deployment-failure bounds
$\epsilon_{\text{tc}}, \epsilon_{\text{dep}},
\epsilon_{\text{clock}}$ from \S\ref{sec:assumptions}, there
exist reductions $B_{\mathrm{gw}}, B_{\mathrm{rec}}, B_H$ such
that
\begin{multline*}
\Adv^{\text{poe}}_{\mathcal{A}}(\lambda) \leq
q_{\mathrm{gw}} \cdot \Adv^{\text{EUF-CMA}}_{\mathrm{Sig}_{\mathrm{gw}}}(B_{\mathrm{gw}})\\
+\, q_{\mathrm{rec}} \cdot \Adv^{\text{EUF-CMA}}_{\mathrm{Sig}_{\mathrm{rec}}}(B_{\mathrm{rec}})\\
+\, q_H^2 \cdot \Adv^{\text{CR}}_H(B_H)\\
+\, \epsilon_{\text{tc}} + \epsilon_{\text{dep}} + \epsilon_{\text{clock}},
\end{multline*}
where $q_{\mathrm{gw}}, q_{\mathrm{rec}}$ are polynomial bounds on
accepted Gateway / Recorder decisions and $q_H$ is a polynomial
bound on hash bindings (the $q_H^2$ factor is the standard
birthday pair enumeration).
If the cryptographic terms are negligible in $\lambda$,
$\Adv^{\text{poe}}_{\mathcal{A}}(\lambda) \leq \negl(\lambda) +
\epsilon_{\text{tc}} + \epsilon_{\text{dep}} +
\epsilon_{\text{clock}}$.
\end{theorem}

\begin{proof}[Proof sketch]
Because $\mathcal{A}$ wins, $\mathsf{ValidatePoE}$ accepts while
$\neg G_k$ for some $k$. By the union bound,
$\Adv^{\text{poe}} \leq \sum_i \Pr[F_i]$. Case analysis maps each
semantic violation to at least one basic failure:
$G_1$ requires a non-oracle root allow ($F_{\text{gw\_sig}}$), a
hash mis-binding ($F_{\text{hash}}$), or an unrecorded effect
($F_{\text{tc}}$); $G_2$, analogously for scope decisions; $G_3$,
a recorded deny-with-effect contradicts $I_3$, so the adversary
must hide the effect from $T^*$ ($F_{\text{tc}}$) or mutate the
record ($F_{\text{rec\_sig}}$ or $F_{\text{hash}}$); $G_4$
requires trace rewriting without signature or hash break
($F_{\text{rec\_sig}}, F_{\text{hash}}$, or $F_{\text{tc}}$);
$G_5$ divergence while $I_{5a}$ holds requires a canonicalization
collision ($F_{\text{hash}}$), an undeclared dependency
($F_{\text{dep}}$), or a clock / commit-order inconsistency
($F_{\text{clock}}$ or $F_{\text{tc}}$). The three reductions
bound the cryptographic failures, and the deployment terms bound
the rest.
\end{proof}

\begin{interpretation}
Thm.~\ref{thm:safety} does not claim PoE is a purely
cryptographic protocol or that A3--A7 are enforced by signatures
and hashes. It claims that, in a deployment satisfying A3--A7,
every PoE-valid but semantically invalid execution must manifest
as a signature forgery, a hash collision, or a quantified
deployment-assumption failure. The substantive content of the
soundness argument lies in the choice of assumption set, the
partition into basic failure events, and the reduction targets
(notably Lemma~\ref{lem:tc}, which converts A5 from a standalone
assumption into a consequence of Effector-exclusive
credentialing); the closure from there is essentially a union
bound.
\end{interpretation}

\subsection{Deterministic replay}
\label{sec:replay}

\begin{theorem}[Replay]
\label{thm:replay}
Under $\WF(C, T, R)$, $\PoEop(C, T, R) = 1$, interpreter
determinism, and A6, $\Replay(C, T, R)$ executed over $\Lin(T)$
yields the same terminal result and persistent state delta as
the original execution.
\end{theorem}

\begin{proofsketch}
Induct on $\Lin(T)$. The base case fixes initial authority and
environment via $C$, $\CR$, $\RL$, and $R$. At the inductive
step, $I_{5a}$ ensures envelope consistency with $R$; A6 ensures
no hidden input enters $\delta$; interpreter determinism ensures
identical inputs produce the same next state. $I_4$ and
$\WF_{\text{dag}}$ rule out valid reordering, deletion, or
fabrication; $\WF_{\text{resource}}$ serializes concurrent
mutations.
\qedhere
\end{proofsketch}

PoE distinguishes three dependency classes: pure deterministic
dependencies, replayable directly; snapshot-able dependencies,
replayable against captured artifacts; and uncaptured
non-deterministic dependencies, outside the replay envelope. The
strength of $G_5$ is inversely proportional to the prevalence of
the third class. Replay fidelity is \emph{internal} consistency
under the recorded context; it is not a claim that the external
world at replay time resembles the world at execution time.

\section{Prime Execution Model and Attestation}
\label{sec:pem}

PoE requires structural separation between \emph{who plans},
\emph{who authorizes}, \emph{who mutates}, and \emph{who records}.
Without such separation, the planner silently becomes the executor
and the trace self-attests. This is a classical
capability-security concern~\cite{saltzer1975, miller2006robust},
instantiated for trajectory-level verification. The Prime Execution
Model (PEM) organizes these responsibilities into six planes
(Table~\ref{tab:planes}). The Enforcement plane instantiates
PEP/PDP discipline, with the signed decision record becoming the
canonical authorization anchor of the trajectory.

\begin{table}[t]
\centering\small
\renewcommand{\arraystretch}{1.1}
\begin{tabular}{@{}l p{2.15in}@{}}
\toprule
\textbf{Plane} & \textbf{Components \& responsibility} \\
\midrule
Governance & Policy Manager, Constraint Registry ---
maintain policy rules and architectural invariants. \\
Planning & Planner --- proposes candidate actions under $C$.
Untrusted with respect to authorization. \\
Enforcement & Gateway / PEP+PDP --- authoritative choke point;
produces the canonical allow/deny record. \\
Effect & Effector --- sole component permitted to produce durable
mutation, and only within the authorized scope. \\
Record & Trace Recorder, Context Store, Contract Registry,
Revocation Log --- append-only, tamper-evident. \\
Observation & Non-authoritative: replay, monitoring, analytics,
validation. Cannot retroactively authorize or rewrite. \\
\bottomrule
\end{tabular}
\caption{The six planes of the Prime Execution Model.}
\label{tab:planes}
\end{table}

\paragraph{Golden Rule.}
\emph{All effectful execution passes through exactly one
authoritative Gateway evaluation.} Path compliance ($G_2$) and
deny semantics ($G_3$) derive from this single architectural
rule, which echoes the Reference Monitor
discipline~\cite{saltzer1975} and is the operational expression of
$I_2$. "Exactly one" refers to one canonical \emph{logical}
decision record, not one physical process.

\subsection{Trace completeness from Effector-exclusive credentialing}

A5 is not a free assumption; PEM supports it.

\begin{lemma}[Trace completeness]
\label{lem:tc}
Assume (1) every persistent mutation inside the declared PoE
effect boundary can be committed only with Effector credentials;
(2) before committing any persistent mutation the Effector must
synchronously emit an event to the Recorder and obtain a sealed
acknowledgment; (3) non-Effector credentials cannot commit
persistent mutations inside the declared boundary; (4) the
Recorder signing root is uncompromised. Then every in-bound
mutation $m$ has a corresponding $e$ with
$\Seal(e) < \Commit(m)$.
\end{lemma}

\begin{proofsketch}
By (1) and (3), every in-bound mutation passes through the
Effector; by (2), it obtains a sealed acknowledgment before
commit; by (4), the acknowledgment cannot be forged. \qedhere
\end{proofsketch}

A credential leak (sealed service account, hardware-attested key)
violates (1) or (3); this case enters the bound of
Thm.~\ref{thm:safety} through T6 and $\epsilon_{\text{tc}}$.

\subsection{Attestation and EAC}
\label{sec:attestation}

The formal artifact of PoE is the \emph{Execution Attestation
Certificate}:
\begin{multline*}
\mathrm{EAC} = \Sign_{\Kgw}(\mathrm{id}(C), \mathrm{root}(T),\\
\mathrm{hash}(R), \mathrm{key\_id}, \mathsf{valid}, \mathrm{timestamp}),
\end{multline*}
issued only when $\PoEop(C, T, R) = 1$, no invariant failed,
replay succeeds when policy requires it, and the attestation
decision is published to the sealed trace before the EAC is
released.

\paragraph{Lifecycle.}
\label{sec:eac-lifecycle}
Gateway keys may rotate; a tamper-evident Key Registry $\KR :
\mathrm{key\_id} \mapsto (pk, \mathrm{validity\_window},
\mathrm{status})$ lets verifiers check historical EACs. If an EAC
should no longer authorize downstream work, a deployment marks it
in an EAC revocation log:
$\RevokeEAC(\mathrm{eac\_id}, \mathrm{reason}, t_{\mathrm{rec}})$.
Historical PoE validity is not rewritten; only current EAC
acceptability changes. Privacy-preserving variants --- selective
disclosure or zero-knowledge EACs proving $\PoEop = 1$ without
revealing the full trace --- are open work.

\paragraph{Proof-driven scheduling.}
The scheduler routes on validated outcomes:
$\mathtt{execution\_error} \to$ retry / escalate;
$\mathtt{contract\_completed} \to$ next workflow node;
$\mathtt{pattern\_detected} \to$ review;
$\mathtt{attestation\_issued} \to$ archive or release downstream
work. The loop in Fig.~\ref{fig:loop} closes when proof validity
authorizes the next step; it does not authorize the next step
merely on successful completion.

\section{Related Work}
\label{sec:related}

\paragraph{Runtime verification and enforcement.}
PoE is closest in spirit to runtime
verification~\cite{leucker2009runtime}. Many of
$I_1,\ldots,I_{5a}$ could be expressed in LTL or monitor-automata
formalisms. PoE differs by binding authorization, effect, sealed
history, replay context, and attestation into one runtime object
--- RV-style checking embedded in a closed loop for proof,
replay, and attestation.

\paragraph{Tamper-evident logs and provenance.}
PoE's history guarantees descend from tamper-evident
logging~\cite{crosby2009secure}, certificate
transparency~\cite{laurie2014ct}, and
provenance~\cite{w3cprov}. PoE additionally binds authorization
and replay into the same object, so the record reconstructs the
execution.

\paragraph{Supply-chain attestations.}
in-toto~\cite{torres2019intoto} and SLSA~\cite{slsa} attest how a
build artifact is produced; PoE targets live agent execution that
is online, effectful, and non-deterministic by default. EACs can
be anchored in SLSA-style provenance for the agent binary and its
skills.

\paragraph{TEEs, consensus, verifiable computation, zkVMs.}
TEEs~\cite{costan2016sgx} attest to enclave integrity;
consensus~\cite{castro1999pbft} agrees on shared state;
verifiable computation~\cite{gennaro2010vc} and
zkVMs~\cite{risczero} prove a function executed per VM semantics.
PoE proves a different object: that a \emph{governed execution}
stayed within contract scope, traversed the authorized path,
produced recorded effects, and is replayable under its recorded
context. PoE composes with zkVMs: a zkVM may prove internal
functional correctness while PoE proves authorization, path, and
effect boundaries.

\paragraph{Capability systems and agent security.}
PEM's authority separation lifts classical capability
discipline~\cite{saltzer1975, miller2006robust} to the
execution-trajectory level. Indirect prompt
injection~\cite{greshake2023injection} establishes that planners
are corruptible; PoE addresses the complementary question of
whether a corrupted planner can produce an attestable execution
outside the authorized boundary. Upstream agent-safety work
provides risk inputs; PoE provides the runtime governance
substrate.

\section{Empirical Evaluation}
\label{sec:eval}

We report preliminary single-node measurements of a PoE prototype
instantiating Gateway, Trace Recorder, Contract Registry,
Revocation Log, Replay Context Store, $\mathsf{ValidatePoE}$, and
EAC issuance. The prototype is implemented in TypeScript
(Node.js~20). Experiments run on Ubuntu~22.04 LTS with Intel Core
i7-12700K (12 cores, 3.6\,GHz), 32\,GB DDR5, and NVMe SSD. Each
measurement is the mean of 1{,}000 trials; reported uncertainties
are 95\% confidence intervals. The baseline $B_0$ is direct tool
invocation without PoE instrumentation.

\paragraph{Overhead.}
Table~\ref{tab:overhead} reports end-to-end latency for three
workloads: a single-capability effectful flow, a five-node
pipeline, and a 50-way parallel batch. PoE adds a roughly
constant $\approx 2.7$\,ms per execution. Gateway decision
(1.1\,ms) and Recorder seal with Merkle hash (0.8\,ms) are the
two dominant components; the remaining $\approx 0.8$\,ms covers
event construction, canonical serialization, and envelope-hash
computation. Gateway decision cost depends on policy complexity:
the reported 1.1\,ms corresponds to a capability-predicate policy
of modest complexity in our prototype and should be expected to
scale with rule count and evaluator implementation. Relative
overhead falls from 65.9\% on a minimal single flow to 4.4\% on
concurrent batch workloads as Gateway and Recorder costs
amortize.

\begin{table}[t]
\centering\footnotesize
\renewcommand{\arraystretch}{1.15}
\setlength{\tabcolsep}{4pt}
\begin{tabular}{@{}l r r r@{}}
\toprule
\textbf{Workload} & \textbf{$B_0$ (ms)} & \textbf{PoE (ms)} &
\textbf{$\Delta$ (\%)} \\
\midrule
single capability & $4.1{\pm}0.3$ & $6.8{\pm}0.4$ & $+65.9$ \\
5-node pipeline & $18.4{\pm}1.1$ & $20.9{\pm}1.2$ & $+13.6$ \\
50-way batch & $61.2{\pm}3.8$ & $63.9{\pm}3.9$ & $+4.4$ \\
\midrule
Gateway decision & --- & $1.1{\pm}0.1$ & --- \\
Recorder + Merkle & --- & $0.8{\pm}0.1$ & --- \\
\bottomrule
\end{tabular}
\caption{End-to-end latency overhead. $n{=}1000$ trials per
cell, 95\% CI.}
\label{tab:overhead}
\end{table}

\paragraph{Trace storage.}
Table~\ref{tab:storage} reports compressed ECES event-stream size
per execution. A standard eight-event trace compresses to
$\approx 1.1$\,KB (gzip at default settings; alternative
compressors and tuning levels will yield different ratios); at
10{,}000 executions/day with a 30-day retention window this is
$\approx 330$\,MB of event-stream storage. Captured inputs held
in $R$ (\S\ref{sec:assumptions})\,---\,tool response bodies,
retrieved document snippets, model-facing context\,---\,are
accounted separately and are workload-dependent; a long-context
LLM agent may add one to two orders of magnitude on top of the
event-stream figure. Hot-path indexed metadata is 0.05\,KB per
execution; full traces are retrieved on demand for replay or
audit.

\begin{table}[t]
\centering\footnotesize
\renewcommand{\arraystretch}{1.15}
\setlength{\tabcolsep}{4pt}
\begin{tabular}{@{}l r r r@{}}
\toprule
\textbf{Profile} & \textbf{Events} & \textbf{Raw (KB)} &
\textbf{Comp. (KB)} \\
\midrule
minimal & 3 & 1.2 & 0.4 \\
standard & 8 & 3.1 & 1.1 \\
complex & 25 & 9.4 & 3.2 \\
indexed metadata only & --- & 0.15 & 0.05 \\
\bottomrule
\end{tabular}
\caption{Compressed ECES event-stream storage per execution
(excluding captured inputs held in $R$).}
\label{tab:storage}
\end{table}

\paragraph{Detection.}
We injected adversarial workloads corresponding to T2 (direct
Effector writes that skip $\mathrm{gateway\_ref}$) and T4
(payload replacement, event reordering, event deletion) across
10{,}000 trials each. $\mathsf{ValidatePoE}$ rejected every
injected trace. No unmodified trace in the control set was
rejected. Rejection is by construction: T4 attacks break
$\mathrm{prev\_event\_hash}$ linkage (violating $I_4$); T2
attacks lack a valid Gateway decision event (violating $I_2$).
These numbers confirm the mechanism, not a security property:
adversaries that successfully forge a Gateway signature, forge a
Recorder seal, or break hash-linking are out of scope of the
prototype evaluation and are addressed only by the bound of
Thm.~\ref{thm:safety}.

\paragraph{Scope of these results.}
The prototype is single-node and implements the invariants and
mechanisms described in this paper. These measurements are
intended to validate mechanism feasibility on commodity hardware,
not to characterize production deployment performance. We
deliberately do not report a comparison against distributed-tracing
frameworks augmented with post-hoc invariant checkers: such a
comparison requires multi-node deployment and workload calibration
against existing telemetry pipelines, and is deferred to a
companion paper along with multi-node measurements, realistic
T1--T6 injection strength, captured-input overhead under
long-context LLM workloads, and replay fidelity over a 30-day
archival window.

\section{Limitations}
\label{sec:limitations}

\paragraph{Valid under contract $\neq$ safe or well-chosen contract.}
$\PoEop = 1$ attests that execution stayed within $C$; it is
silent on whether $C$ was well chosen, and on whether an
authorized action is wise. A contract that authorizes harmful
actions yields a PoE-valid harmful execution. Contract authoring,
review, and revocation are upstream problems.

\paragraph{Replayable $\neq$ reproducible in the world.}
Thm.~\ref{thm:replay} guarantees replay under captured inputs,
declared dependencies, and interpreter determinism. It does not
guarantee that external data sources, tools, or model versions at
replay time match those at execution time.

\paragraph{Single-Gateway scope; composition open.}
Thm.~\ref{thm:safety} is explicit about single-logical-Gateway
deployments. Threshold Gateways require a quorum $I_2$, quorum
certificates, view-change / equivocation handling, and a BFT
safety proof. Concurrent contracts sharing a resource require
cross-contract serialization semantics beyond this paper.

\paragraph{Planner compromise is out of scope.}
PoE confines what a compromised planner can do within an issued
contract; it does not prevent compromise. Upstream defenses ---
prompt-injection hardening, interpretability, and
contract-origination controls --- are orthogonal.

\section{Conclusion}

PoE proposes a primitive:
$\text{Execution} \to \text{Proof} \to \text{Attestation} \to
\text{Scheduling}$. Its contribution is not faster computation or
alternative consensus. It is a reframing of what constitutes
\emph{work}. Legacy systems measure work by computation consumed
or state agreed upon. PoE treats it as \emph{attestable
execution under contract}. The formal results separate two
classes of failure: cryptographic (signature unforgeability,
collision resistance) and deployment (trace completeness,
dependency declaration, clock monotonicity). Making both
explicit is what moves PoE from architectural claim to auditable
systems-security claim. Preliminary single-node measurements
(\S\ref{sec:eval}) show $\approx 2.7$\,ms Gateway-plus-sealing
overhead amortizing to 4.4\% on concurrent batches, and
compressed traces around 1\,KB per standard execution ---
indicating the mechanism is practical on commodity hardware at
enterprise workload scales. Next: extend the model to threshold
Gateways and concurrent contract composition, and report
multi-node prototype measurements with realistic captured-input
overheads under long-context LLM workloads in a companion paper.

\section*{Acknowledgments}
We thank the AlphaBitCore engineering team for discussions that
shaped the architectural model. Remaining errors are our own.

\appendix

\section{Witness Constructions for Theorem~\ref{thm:independence}}
\label{app:witnesses}

Each $S_k$ satisfies every invariant except $I_k$ (or, for $S_5$,
satisfies $I_{5a}$ but violates $G_5$ via A6 failure).

\begin{center}\small
\renewcommand{\arraystretch}{1.2}
\begin{tabular}{@{}l p{1.4in} l l@{}}
\toprule
\textbf{$S_k$} & \textbf{Construction} & \textbf{Violates $I$} &
\textbf{Violates $G$} \\
\midrule
$S_1$ & No root allow event $e_a$; descendants locally sealed. & $I_1$ & $G_1$ \\
$S_2$ & Side-channel effector writes with $\mathrm{gateway\_ref}$ pointing to a non-canonical decision. & $I_2$ & $G_2$ \\
$S_3$ & Gateway emits deny $d$; effector writes a durable cache entry on the deny branch. & $I_3$ & $G_3$ \\
$S_4$ & Recorder supports post-hoc event replacement with regenerated local seal. & $I_4$ & $G_4$ \\
$S_5$ & $\delta_e$ reads undeclared hidden input $z$; envelope check passes on declared fields. & ($I_{5a}$ passes) & $G_5$ (A6) \\
\bottomrule
\end{tabular}
\end{center}

\noindent $S_5$ demonstrates why A6 is a named assumption rather
than an invariant: the validator cannot detect undeclared
dependencies from the trace alone, so dependency-declaration
completeness enters Thm.~\ref{thm:safety} as
$\epsilon_{\text{dep}}$.

\bibliographystyle{plainnat}
\bibliography{refs_v11}

\end{document}